\documentclass[a4paper,11pt]{article}
\usepackage[utf8]{inputenc}
\usepackage{amsthm}
\usepackage{amsmath}
\usepackage{amssymb}
\usepackage{caption}
\usepackage{cite}
\usepackage{tikz}
\usepackage{subcaption}
\usepackage[labelformat=simple]{subcaption}
\usetikzlibrary{calc}
\usepackage{color,url}

\usepackage{booktabs}
\usepackage{subcaption}
\usepackage{pgfplots}
\pgfplotsset{compat=1.18}

\usepackage{changes}

\usepackage{hyperref}

\hypersetup{colorlinks=true, linkcolor=blue, citecolor=blue, urlcolor=blue}

\usepackage{enumitem}
\setlist[enumerate]{noitemsep,topsep=3pt}
\setlist[itemize]{noitemsep,topsep=3pt}

\theoremstyle{definition}
\newtheorem{defin}{Definition}

\theoremstyle{plain}

\newcommand{\df}{\mathrm{df}}

\newcommand{\cO}{{\cal O}}
\newcommand{\cT}{{\cal T}}

 \newtheorem{remark}{Remark}
 \newtheorem{statement}{Statement}

 \newtheorem{lemma}[defin]{Lemma}
 \newtheorem{theorem}[defin]{Theorem}

 \theoremstyle{definition}

\begin{document}

\title{Vertices that belong to every minimum dominating set of a graph and their connection with transportation sharing systems with study cases in Campo de Gibraltar area}

\author{Mohammad Farhan$^{a,}$\thanks{\texttt{mohammad.farhan@uca.es}}
\and Dorota Kuziak$^{a,}$\thanks{\texttt{dorota.kuziak@uca.es}}
\and
Iztok Peterin$^{b,}$\thanks{\texttt{iztok.peterin@um.si}}
\and Ismael G. Yero$^{c,}$\thanks{\texttt{ismael.gonzalez@uca.es}}
}

\maketitle

\begin{center}
$^a$ Departamento de Estad\'istica e Investigaci\'on Operativa, Universidad de C\'adiz, Algeciras Campus, Spain \\

\medskip
$^b$ Faculty of Electrical Engineering and Computer Science, University of Maribor, Slovenia\\

\medskip
$^c$ Departamento de Matem\'{a}ticas, Universidad de C\'adiz, Algeciras Campus, Spain
\end{center}

\maketitle

\begin{abstract}
This study addresses a theoretical model regarding equity and accessibility challenges in designing shared transportation systems (such as micro-mobility networks) by applying graph-vertex domination setting. The work focuses on identifying dominating forced vertices, that represent nodes belonging to every dominating set of a graph of the smallest possible cardinality, and which correspond to critical, non-negotiable station locations essential for maintaining system efficiency and coverage.

From a theoretical perspective, in the paper it is first demonstrated that determining whether a given vertex is a dominating forced vertex is co-NP-hard, establishing the computational infeasibility of exact identification in large networks. To analyze graph structures, sharp theoretical bounds on the maximum number of dominating forced vertices are established, proving that their count is bounded above by one-third of the order of the graph, and provide complete structural characterizations for graphs achieving this bound, as well as, trees with no dominating forced vertices.

To overcome computational limits in practical urban settings, the study uses an iterated greedy metaheuristic framework to generate minimal dominating sets and approximate critical forced nodes based on their appearance frequency across iterations. The methodology is validated on strong grid graphs and applied to real-world road network models of Algeciras and La L\'inea de la Concepci\'on, two cities in the area of Campo de Gibraltar, Spain, which successfully pinpoints candidate location points for scooter-sharing stations across both cities.
\end{abstract}


\noindent
{\bf AMS Subj.\ Class.\ (2020)}: 05C12

\section{Introduction}

The global expansion of bicycle (scooters or related transport methods) sharing systems along several cities in the world has revolutionized the urban mobility of people. The usefulness of these systems is questionless, mainly due to its high influence into the development of less polluted and less crowded cities, which is undoubtedly positive for the improvement of our society.  The success of these systems does not depend solely on the fleet size used (number of bikes, scooters, etc.), but also on the strategic placement of stations along a determined city or region. Urban planners face several fundamental challenges while designing such a sharing systems for a determined region, like for instance the following ones.

\begin{itemize}
    \item Logistical costs, where operators must use trucks to manually move bikes from full stations to empty ones, which is indeed one of the most expensive parts of operations.
    \item Demand prediction, which involves the sudden weather changes or events that can disrupt planned routes, requiring advanced predictive algorithms to maintain service levels.
    \item Equity and accessibility, that addresses the problem about how to guarantee that every citizen has a station within a walking distance while minimizing infrastructure costs.
\end{itemize}

Each one of these challenges has their own approaches for solutions and there is a wide range of investigations covering these and many other issues for generating or improving different sharing systems. In our investigation, we aim to focus in the equity or accessibility issue. Notably, one possible approach to it lies in the area of graph theory, and specifically in the area of domination models, which is in turn, a classical topic in the theoretical sense.

In graph theory, a city is usually modeled as a set of nodes (intersections or points of interest) connected by edges (streets). A dominating set is a subset of nodes where stations are placed in such a way that every node in the city is either a station or directly adjacent to one, which indeed relates to the equity or accessibility of the sharing system. The objective for system planners aims to identify the smallest possible dominating set that would serve as location points for the bicycles (or similar transportation units). Such a model (with the smallest possible number of locations) clearly maximizes the coverage of the users in the city, with a lowest possible investment. Notice that, we yet do not consider the important detail regarding the number of transportation units that are necessary at each location regarding the number of possible users of the system, nor we analyze the problem of relocating such units, if this would be required. Namely, we mainly focus on establishing the ``initial state'' of placing the transportation devices so that the system starts functioning.

This basic idea of dominating points is clearly very useful, although it lacks of several improvements, study of particular situations of a city, possible traffic situations (streets where bikes cannot be placed, etc.). However, there are a few details that might be taken into account. For example, the problem of finding the smallest possible dominating set in a graph is a very well known hard problem (called NP-hard, due to the complexity class where it belongs), see \cite{book6}. This makes that a system planner might be in trouble while finding a solution for identifying locations in a city for placing the transportation units, if the city is ``big enough''. Consequently, several times, only some approximated results can be obtained, but at least, such approximated results frequently works very efficiently, and a best possible solution is not strictly needed.

On the other hand, in domination theory, a vertex is considered dominating critical if it belongs to every possible minimum dominating set of a graph. Such vertex can be also called ``forced'' according to another terminologies. Now, for a sharing system planner, these points are non-negotiable. If a vertex is dominating critical, it means that without a station at that exact location, the system’s efficiency (equity and accessibility) collapses or the cost of covering adjacent areas increases too much. These points ``intuitively'' coincide with vertices adjacent to other vertices that have no more neighbors (called leaves), or to vertices with a ``high'' degree (number of neighbors). The identification of such vertices is then crucial for designing any efficient sharing system.

In our investigation we precisely focus into the identification of such dominating critical vertices in a given graph. Unfortunately, for the sharing system planners, we first precisely show that finding such vertices in a graph is not an easy task. That is, showing that a given vertex of a graph belongs to every minimum dominating set of such graph is a problem included in the NP-hard complexity class. Further on, we develop a greedy algorithm that finds these critical (or forced) vertices in a graph, and make a particular study for possible transportation sharing system in the two main cities from the area of ``Campo de Gibraltar'', Spain, which are Algeciras and La L\'inea de la Concepci\'on. We first develop some computational study to detect these dominating critical or forced vertices in such cities, as well as, develop some computational implementations on random graphs in order to evaluate the efficiency of the algorithm and gives some conclusions regarding empirical bounds for the number of these dominating critical or forced vertices in a given graph.

\subsection{Formal concepts}

The topic of domination in graphs is a classical research area in graph theory, and any interested person can easily find a large amount of information about it. The number of ongoing investigations cover a very wide range of different directions, arising in applied situations, as well as, inside the theory itself. To see the vast information on these facts, the reader can consult some of the books \cite{book1,book2,book3,book4,book5,book6,book7,book8}. The classical concepts on domination are as follows. Given a graph $G$, a set of vertices $D\subseteq V(G)$ is a \textit{dominating set} of $G$, if every vertex not in $D$ has a neighbor in $D$. The \textit{domination number} of $G$ is the minimum cardinality of a dominating set of $G$, and is usually denoted by $\gamma(G)$. A dominating set of $G$ of cardinality $\gamma(G)$ is called a \textit{minimum dominating set} (of $G$), or a $\gamma(G)$-\textit{set}. In addition, a dominating set is called \textit{minimal} if does not properly contain a dominating set of smallest cardinality.

In closed relationship with these concepts there are lots of other structures and parameters in the literature, each of them interesting by itself. Many of these domination parameters are nowadays very well studied, and clearly, the classical ones of domination number and dominating sets are probably the most studied. However, there are still numerous open questions about it that are of interest for the research. In our investigation, we focus on considering those vertices that must belong to every minimum dominating set and their usefulness for designing bicycle sharing systems. This theoretical notion was first considered in \cite{Mynhardt}, and somehow further rediscovered in \cite{Bouquet}. Some other contributions about this class of vertices appeared in \cite{Klostermeyer}, and also, very recently in \cite{Ziemann}. Apart from these works, not that large specifically knowledge about such vertices is known, although the reader can find several parallel results regarding those vertices that belong to all or to no minimum dominating set with an added property. That is, one can find such similar studies for variations of the classical domination concept including total domination, double domination, paired domination, and others. Among them, we remark \cite{Blidia,Blidia2,Chen,Chen2,Cockayne,Gunther,Henning,Henning-Plum,Meddah,Meddah2}. In addition, the reader can see, for instance, some studies considering graphs with unique minimum dominating sets (or related concept) that clearly induce the existence of dominating forced vertices. Specifically, \cite{Gunther} initiated such topic, as well as, some other further contributions like \cite{Fischermann,Fraboni,Hedetniemi,Zhao} are closely related.

A vertex $v\in V(G)$ is a \textit{dominating forced vertex} of $G$ if it belongs to every minimum dominating set (or $\gamma(G)$-set) of $G$. The number of dominating forced vertices of a graph $G$ is denoted by $\df(G)$. We must remark that the set of dominating forced vertices of a graph $G$ were already represented as ${\rm core}(G)$ in \cite{Bouquet}. In this sense, clearly $\df(G)=|{\rm core}(G)|$. Among the main results that can be remarked regarding this topic, in \cite{Mynhardt}, it was proved that for any tree $T$ and any vertex $v\in V(T)$, it holds that $v$ is contained in every minimum dominating set $T$ if and only if the number of children $w$ of $v$ whose subtree $T_w$  contains a path $P$ with length $\ell\equiv 0 \pmod{3}$ is greater than or equal to $2$. In \cite{Klostermeyer}, the authors proved that the non existence of dominating forced vertices in a tree $T$ is a necessary condition for $T$, to satisfy that the eternal eviction number equals the classical domination number. On the other hand, the work \cite{Bouquet} considered some computational issues concerning finding these vertices in interval graphs. Specifically, there was established that the problem of verifying whether a vertex is dominating forced can be solved in time $\mathcal{O}(|V(G)| + |E(G)|)$ for any interval graph $G$, which results in a quadratic time algorithm for such a problem in this graph class. Finally, in \cite{Ziemann}, a linear time algorithm for determining the sets of vertices that belong to all (i.e., dominating forced vertices), to some, or to no minimum dominating sets of a tree was provided.

It is readily observed that it might happen that $\df(G)=0$ for some graphs $G$, i.e., $G$ has no dominating forced vertices. Trivial examples of this situation are for instance the complete graph $K_n$ or the cycle $C_n$. On the other hand, it must clearly happen that $\df(G)$ is at most the domination number of $G$. Clearly, this situation only happen in the case that $G$ will have a unique minimum dominating set. As a consequence of these comments, for any graph $G$ it happens that
\begin{equation}
\label{eq:triv-bounds}
0\le \df(G)\le \gamma(G).
\end{equation}

The idea of considering ``forced vertices'' in connection with other structures in graph theory is a relatively common one. For example, there also exists this notion for maximal independent sets (see \cite{Boros}); and for the metric bases (see \cite{Hakanen,Hakanen1}).

In the article \cite{Bouquet}, authors considered not only the vertices that belong to every minimum dominating set, but also, the vertices that do not belong to any of such sets, and the ones that belong to some of them. Their results were mainly focused into giving conditions for a vertex to belong to one of such classifications for general graphs, as well as, for some particular classes. In our exposition, we go further away than this, and try to use them to solve some practical situations in urban mobility.

\subsection{Other terminology and notation}

Let $G$ be a graph, let $S\subset V(G)$, and let $u\in S$. We use $G[S]$ for a subgraph of $G$ induced by vertices from $S$. The set of neighbors of $u$ is called the \textit{open neighborhood} of $u\in V(G)$ and is denoted by $N_G(u)$. By $N_G[v]$ we denote the \textit{closed neighborhood} of $v$ which is $N_G(v)\cup \{v\}$ and further we use $N_G[S]$ for the set of all vertices in $S$ together with vertices adjacent to some vertices from $S$. The vertex $v\notin S$ is a \textit{private neighbor} of $u$ with respect to $S$ if $N_G(v)\cap S= \{u\}$. The set of private neighbors of $u$ with respect to $S$ is denoted by $pn(v,S)$.

As usual is ${\rm deg}(v)=|N_G[v]|$ the \textit{degree} of $v\in V(G)$. A vertex $v$ with ${\rm deg}(v)=0$ is an \textit{isolated vertex} of $G$. A vertex $v$ with ${\rm deg}(v)=1$ is a \textit{leaf} of $G$ and the vertex adjacent to a leaf $x$ is called the \textit{support} of $x$. Further, we distinguish \textit{strong support vertices}, that are adjacent to more than one leaf, and \textit{weak support vertices} with only one leaf neighbor. Notice that, any strong support vertex is a dominating forced vertex.

Given a graph $G$ with vertex set $V(G)=\{u_1,\dots,u_n\}$ and any graph $H$, the \textit{corona graph} $G\odot H$ is the graph obtained from one copy of $G$ and $n$ copies of $H$, say $H_1$, $\dots$, $H_{n}$, by adding edges between the vertex $u_i\in V(G)$ and every vertex $v\in V(H_i)$.

If a graph $G$ has an isolated vertex, then clearly such vertex is a dominating forced vertex. In this sense, throughout our whole exposition we consider graphs having no isolated vertices.

\section{Some related studies}

The concepts of dominating sets and domination number are nowadays some of the most studied topics in graph theory. To confirm this, we just suggest the reader to check the several already cited books \cite{book1,book2,book3,book4,book5,book6,book7,book8}. Our main goal in this investigation is not to study theoretical aspects of dominating sets. Therefore, we center our attention in this section into related studies that consider ``dominating properties'' of graphs applied to transportation sharing systems or to the identification of location points in a network that would keep some domination features. Further, we only consider studies that are related to our investigation and do not try to cover all the literature from the huge area of transportation sharing systems.

In \cite{related1}, the authors presented a study in which they consider efficient resource allocation in urban environments by means of some graph-theoretic approaches. To this end, optimizing resource allocation is understood by modeling urban zones as interval graphs and further on identifying minimum dominating sets to ensure complete coverage with minimal costs. In their investigation, the authors applied a linear-time algorithm to a case study involving an urban corridor with 6 zones (to better model with interval graphs). They demonstrated that a full spatial coverage to the urban system mobility can be achieved using only three zones in the design. The results presented also discussed the efficiency, scalability of their model, and some ideas are given for its integration into practical situations for urban infrastructure design.

In \cite{related2}, another urban transportation system is studied. It is related to facing challenges in balancing taxi supply and demand, with some emphasis in less congested areas. In the investigation, an approach related to dominating sets of graphs is used to identify strategic points (which they call local hot spots) throughout the city, based on historical taxi request data and the road network. This idea is clearly useful and allows taxi drivers to find nearby customers after a trip. This contributes to reducing waiting times and increasing efficiency. In order to proceed with their approach, the investigation consider some sort of identifies neighboring nodes within a $k$-hop distance, which allows then to introduce the concept of $k$-hop dominating set. Namely, it selects nodes with the highest demand and their neighbors within a $k$-hop range. The idea is that, once a $k$-hop dominating set is detected, then the system assigns requests to the nearest drivers and redirects them to the $k$-hop dominating set nodes if they are not busy so that they might be rapidly ``accessible'' when required. The investigation contains some experimental results, using data from the New York City road network and its taxi database. As a consequence, there is a proposal to improve taxi service efficiency in such a city.

The work \cite{related3} considers the problem of location for bike stations regarding the potential demand around the locations. Aspect like population, commercial activities or public transport stations (bus, train, etc.) have to be taken into account. The study proposes the use of a Geographical Information System (GIS)-based method in order to calculate the potential distribution for the demand of trips, the location of stations and to determine the station capacities, that are necessary to establish the characteristics of the demand at each station. The investigation focuses into two directions: minimizing impedance and maximizing coverage of the system location. Although not explicitly mentioned, those concepts clearly relates to dominating sets of graphs. Minimizing impedance is traduced to station location optimized by minimizing the distance between users and location points, while maximizing coverage focuses into optimizing the total number of users covered within a particular distance radius (in the work the length 200 meters is established). The work considers some particular emphasis in the Spanish city of Madrid.

In \cite{related4}, the problem of modeling and optimizing facility locations in road networks using directed graphs (digraphs) is addressed in connection with the study of $k$-dominating sets in graphs (set of vertices in which each vertex dominates at least $k\ge 1$ vertices). A new general concept of a reachability digraph associated with a road network is introduced, in order to model the placement of refueling facilities in such road networks. The investigation contains some greedy heuristics and experimental results for searching $k$-dominating sets in digraphs (specially in their newly introduced reachability digraphs). In addition, the authors presented some results of  computational experiments that were carried out for a case study of road networks in the West Midlands (UK).

\medskip
To the best of our knowledge, these works mentioned above are probably the only ones in which domination topics in graphs are relatively clearly presented and used to solve or to model the problem of locations of stations in transportation sharing systems. In addition, we remark that none of them have considered the idea of detecting ``crucial location points'' in the system in the sense of dominating forced vertices. Thus, our investigation is the first one of this type, which indeed shows its novelty.

\section{Complexity issues}
\label{sec:complex}

In this section we are focused on the following decision problem, to first prove the difficulty of the situation that we want to address. Namely, while it would be clearly desirable to detect those key points in a transportation sharing system, we show next that such points are indeed not ``easy'' to find. We consider the following problem.

\begin{center}
\fbox{
	\parbox{0.9\textwidth}{
		\textsc{\bf Dominating Forced Vertex Problem (DFV problem)} \\
		\textit{Instance}: A simple, non-directed graph $G$ and a vertex $v\in V(G)$.\\
		\textit{Question}: Is $v$ a dominating forced vertex?}}
\end{center}

We recall that the problem above can be efficiently solved for some families of graphs like trees and interval graphs, as already shown in \cite{Ziemann} and \cite{Bouquet}, respectively. However, in general, this efficiency does not remain further as we next argue. In order to show that the computational complexity of the problem above is co-NP-complete in general, we need some extra information. First of all, we use a reduction from the well known 3-SAT problem, and some techniques that were already used in the classical proof for the NP-completeness of the decision problem regarding the domination number of graphs appeared in \cite{Garey} (see also the book \cite[page 34]{book6}). That is, the next problem.

\begin{center}
\fbox{
	\parbox{0.9\textwidth}{
		\textsc{\bf Dominating Set Problem (DomSet problem)} \\
		\textit{Instance}: A simple, non-directed graph $G$ and an integer $k\ge 1$.\\
		\textit{Question}: Does $G$ has a dominating of cardinality at most $k$?}}
\end{center}

Consider $n$ binary variables $x_1,x_2,\dots,x_n$ that can take two values: True (T) or False (F). Let $X=\{x_1,x_2,\dots,x_n\}$ and $\overline{X}=\{\overline{x_1}, \overline{x_2}, \dots, \overline{x_n}\}$ be the sets of literals concerning the variables, where $x_i$ represents a positive value of the variable and $\overline{x_i}$ its negative value (if $x_i$ has value true, then $\overline{x_i}$ has value false and vice versa). Recall that the 3-satisfiability problem (3-SAT problem), over a collection (also called formula) of 3-literals sets (called clauses) of the set $X\cup \overline{X}$,  asks for the existence of an assignment of values true (T) or false (F) to the variables $x_1,x_2,\dots,x_n$, in such a way that in each of the clauses there is at least one literal with value true (T) - which is also said that each clause is satisfied, and thus the formula is satisfied as well. In this sense, it is also said that the formula is satisfiable.

In order to prove that the DomSet problem is NP-complete, a reduction from the 3-SAT problem was used in \cite{Garey} as follows.
\begin{itemize}
    \item For any variable $x_i$, a triangle (a complete graph $K_3$) is created and its vertices denoted $T_i,F_i,v_i$.
    \item For each clause $C_j$ a single vertex $C_j$ is created.
    \item For each clause $C_j\subset X\cup \overline{X}$, if $x_i\in C_j$, then add the edge $C_jT_i$. On the contrary, if $\overline{x_i}\in C_j$, then add the edge $C_jF_i$.
\end{itemize}
For a given formula $\mathcal{F}=\{C_1,\dots,C_m\}$ of 3-literals sets over the variables $x_1,x_2,\dots,x_n$, a graph $G_F$ is constructed by using the rules above. It was proved in \cite{Garey}, that $\mathcal{F}$ is satisfiable if and only if the domination number of $G_F$ equals $n$. A representative example of a graph $G_F$ is drawn in Figure \ref{fig:G_F}, obtained from the formula $\mathcal{F}=\{C_1=\{x_1,\overline{x_2},x_3\}, C_2=\{x_2,\overline{x_3},x_4\},C_3=\{x_3,x_4,x_5\}, C_4=\{\overline{x_3},\overline{x_4},\overline{x_5}\}\}$.

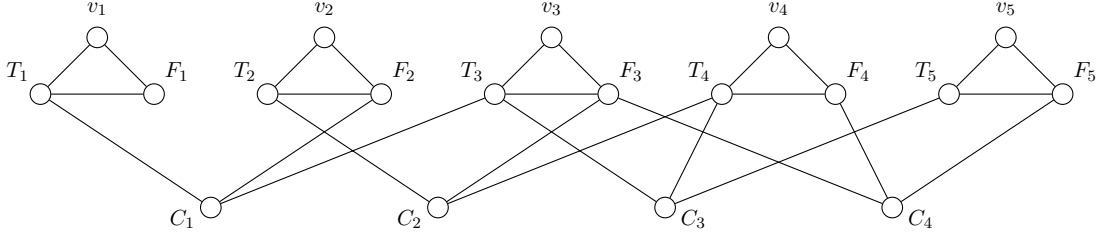
\begin{figure}[h]
\centering
\begin{tikzpicture}[scale=.75, transform shape]
\node [draw, shape=circle,fill=white] (v3) at  (0,3) {};
\node at (0,3.5) {$v_3$};
\node [draw, shape=circle,fill=white] (v4) at  (4,3) {};
\node at (4,3.5) {$v_4$};
\node [draw, shape=circle,fill=white] (v5) at  (8,3) {};
\node at (8,3.5) {$v_5$};
\node [draw, shape=circle,fill=white] (v2) at  (-4,3) {};
\node at (-4,3.5) {$v_2$};
\node [draw, shape=circle,fill=white] (v1) at  (-8,3) {};
\node at (-8,3.5) {$v_1$};
\node [draw, shape=circle,fill=white] (f3) at  (1,2) {};
\node at (1.4,2.4) {$F_3$};
\node [draw, shape=circle,fill=white] (t4) at  (3,2) {};
\node at (2.6,2.4) {$T_4$};
\node [draw, shape=circle,fill=white] (f4) at  (5,2) {};
\node at (5.4,2.4) {$F_4$};
\node [draw, shape=circle,fill=white] (t5) at  (7,2) {};
\node at (6.6,2.4) {$T_5$};
\node [draw, shape=circle,fill=white] (f5) at  (9,2) {};
\node at (9.4,2.4) {$F_5$};
\node [draw, shape=circle,fill=white] (t3) at  (-1,2) {};
\node at (-1.4,2.4) {$T_3$};
\node [draw, shape=circle,fill=white] (f2) at  (-3,2) {};
\node at (-2.6,2.4) {$F_2$};
\node [draw, shape=circle,fill=white] (t2) at  (-5,2) {};
\node at (-5.4,2.4) {$T_2$};
\node [draw, shape=circle,fill=white] (f1) at  (-7,2) {};
\node at (-6.6,2.4) {$F_1$};
\node [draw, shape=circle,fill=white] (t1) at  (-9,2) {};
\node at (-9.4,2.4) {$T_1$};

\node [draw, shape=circle,fill=white] (c3) at  (2,0) {};
\node at (2.5,-0.2) {$C_3$};
\node [draw, shape=circle,fill=white] (c4) at  (6,0) {};
\node at (6.5,-0.2) {$C_4$};
\node [draw, shape=circle,fill=white] (c2) at  (-2,0) {};
\node at (-2.5,-0.2) {$C_2$};
\node [draw, shape=circle,fill=white] (c1) at  (-6,0) {};
\node at (-6.5,-0.2) {$C_1$};

\draw(f1)--(v1)--(t1)--(f1);
\draw(f2)--(v2)--(t2)--(f2);
\draw(f3)--(v3)--(t3)--(f3);
\draw(f4)--(v4)--(t4)--(f4);
\draw(f5)--(v5)--(t5)--(f5);

\draw(t1)--(c1)--(f2);
\draw(c1)--(t3)--(c3)--(t4)--(c2)--(t2);
\draw(c2)--(f3)--(c4)--(f5);
\draw(c3)--(t5);
\draw(c4)--(f4);
\end{tikzpicture}
\caption{The graph $G_F$ constructed from the formula $\mathcal{F}=\{C_1=\{x_1,\overline{x_2},x_3\}, C_2=\{x_2,\overline{x_3},x_4\},C_3=\{x_3,x_4,x_5\}, C_4=\{\overline{x_3},\overline{x_4},\overline{x_5}\}\}$.}\label{fig:G_F}
\end{figure}

Now, in order to consider our DFV problem, we adapt the construction described above as follows. We begin also with a formula $\mathcal{F}=\{C_1,\dots,C_m\}$ of 3-literals sets over the variables $x_1,x_2,\dots,x_n$, and construct a graph $G_F$ as described above. Next, to construct a graph $G'_F$, we add a graph $K_2$ with vertex vertex set  $\{z,z'\}$ and all the edges $z C_j$ for every $C_j\in \mathcal{F}$. See Figure \ref{fig:G_F-prime} for an example of a graph $G_F'$ obtained from the graph $G_F$ drawn in Figure \ref{fig:G_F}.

\begin{figure}[h]
\centering
\begin{tikzpicture}[scale=.75, transform shape]
\node [draw, shape=circle,fill=white] (v3) at  (0,3) {};
\node at (0,3.5) {$v_3$};
\node [draw, shape=circle,fill=white] (v4) at  (4,3) {};
\node at (4,3.5) {$v_4$};
\node [draw, shape=circle,fill=white] (v5) at  (8,3) {};
\node at (8,3.5) {$v_5$};
\node [draw, shape=circle,fill=white] (v2) at  (-4,3) {};
\node at (-4,3.5) {$v_2$};
\node [draw, shape=circle,fill=white] (v1) at  (-8,3) {};
\node at (-8,3.5) {$v_1$};
\node [draw, shape=circle,fill=white] (f3) at  (1,2) {};
\node at (1.4,2.4) {$F_3$};
\node [draw, shape=circle,fill=white] (t4) at  (3,2) {};
\node at (2.6,2.4) {$T_4$};
\node [draw, shape=circle,fill=white] (f4) at  (5,2) {};
\node at (5.4,2.4) {$F_4$};
\node [draw, shape=circle,fill=white] (t5) at  (7,2) {};
\node at (6.6,2.4) {$T_5$};
\node [draw, shape=circle,fill=white] (f5) at  (9,2) {};
\node at (9.4,2.4) {$F_5$};
\node [draw, shape=circle,fill=white] (t3) at  (-1,2) {};
\node at (-1.4,2.4) {$T_3$};
\node [draw, shape=circle,fill=white] (f2) at  (-3,2) {};
\node at (-2.6,2.4) {$F_2$};
\node [draw, shape=circle,fill=white] (t2) at  (-5,2) {};
\node at (-5.4,2.4) {$T_2$};
\node [draw, shape=circle,fill=white] (f1) at  (-7,2) {};
\node at (-6.6,2.4) {$F_1$};
\node [draw, shape=circle,fill=white] (t1) at  (-9,2) {};
\node at (-9.4,2.4) {$T_1$};

\node [draw, shape=circle,fill=white] (c3) at  (2,0) {};
\node at (2.5,-0.2) {$C_3$};
\node [draw, shape=circle,fill=white] (c4) at  (6,0) {};
\node at (6.5,-0.2) {$C_4$};
\node [draw, shape=circle,fill=white] (c2) at  (-2,0) {};
\node at (-2.5,-0.2) {$C_2$};
\node [draw, shape=circle,fill=white] (c1) at  (-6,0) {};
\node at (-6.5,-0.2) {$C_1$};

\node [draw, shape=circle,fill=white] (z) at  (0,-1.5) {};
\node at (0,-1) {$z$};
\node [draw, shape=circle,fill=white] (z') at  (3,-1.5) {};
\node at (3.5,-1.3) {$z'$};

\draw(c1)--(z)--(c2);
\draw(c3)--(z)--(c4);
\draw(z')--(z);

\draw(f1)--(v1)--(t1)--(f1);
\draw(f2)--(v2)--(t2)--(f2);
\draw(f3)--(v3)--(t3)--(f3);
\draw(f4)--(v4)--(t4)--(f4);
\draw(f5)--(v5)--(t5)--(f5);

\draw(t1)--(c1)--(f2);
\draw(c1)--(t3)--(c3)--(t4)--(c2)--(t2);
\draw(c2)--(f3)--(c4)--(f5);
\draw(c3)--(t5);
\draw(c4)--(f4);
\end{tikzpicture}
\caption{The graph $G'_F$ constructed from the graph $G_F$ drawn in Figure \ref{fig:G_F}.}\label{fig:G_F-prime}
\end{figure}
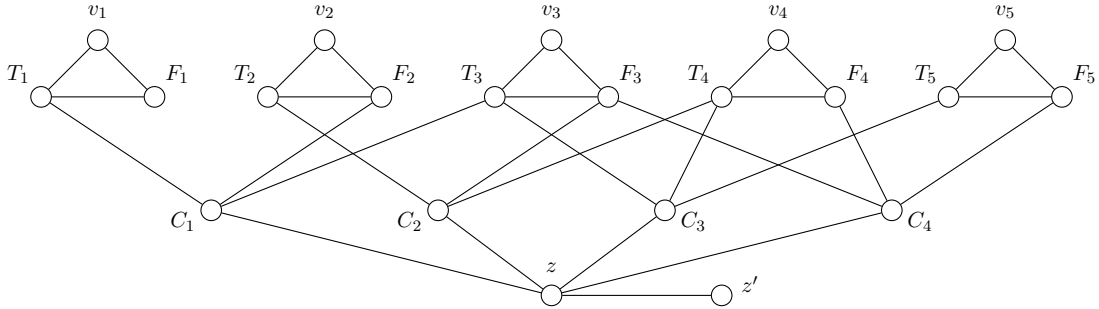

We now see a couple of properties of $G_F'$. Since each vertex $v_i$ from each triangle corresponding to a variable $x_i$ must be dominated by some vertex inside the triangle, as well as, the vertex $z'$ needs to be also dominated, we obtain that each dominating set of $G_F'$ must contain at least $n+1$ vertices. On the other hand, it can be readily observed that  any set of vertices of $G_F'$ formed by exactly one vertex of each triangle corresponding to a variable $x_i$, together with the vertex $z$ is a dominating set of $G_F'$. Therefore, the following observation is true.

\begin{remark}\label{rem:G_F'}
For every graph $G_F'$ we have $\gamma(G_F')=n+1$.
\end{remark}

The next result is the core of our complexity study on our DFV problem.

\begin{lemma}\label{lem:complex}
Given a formula $\mathcal{F}=\{C_1,\dots,C_m\}$ over the variables $x_1,\dots,x_n$, the vertex $z$ of $G_F'$ is not a dominating forced vertex of $G_F'$ if and only if $\mathcal{F}$ is satisfiable
\end{lemma}

\begin{proof}
From Remark \ref{rem:G_F'}, we know that $\gamma(G_F')=n+1$. Assume that $z$ is not a dominating forced vertex of $G_F'$, which means there is a dominating set $S$ of $G_F'$ with cardinality $n+1$ such that $z\notin S$ (notice that each triangle of $G_F'$ corresponding to a variable contains exactly one vertex of $S$). Also, it must happen that $z'\in S$. Moreover, the vertices $C_1,\dots,C_m$ must be dominated by some vertices $T_i$ or $F_i$ from the triangles. In addition, we may also assume that $S$ has no vertex of type $v_\ell$. In such a case ($v_\ell\in S$ for some $\ell$), then we can simply exchange $v_\ell$ with either $T_\ell$ or $F_\ell$, and changed $S$ is still a dominating set.

Now, if there is a vertex $T_i\in S$ for some $i$, then we set the variable $x_i$ as true (T). On the contrary, if $F_i\in S$, then we set $x_i$ as false (F). Now, similar arguments as the ones used in \cite{Garey}, show that such assignment satisfies the formula $\mathcal{F}$. That is, since each $C_j=\{u,v,w\}\ni x_i$ must be dominated by $S$, such $C_j$ must be adjacent to a vertex $T_i$ (in which case $x_i$ has been assigned true) or to a vertex $F_i$ (in which case $x_i$ has been assigned false).

On the other hand, if $\mathcal{F}$ is satisfiable, then we can always construct a dominating set $S'$ of $G_F'$ of cardinality $n+1$, containing exactly one vertex of each triangle of $G_F'$ corresponding to a variable together with the vertex $z'$. Notice that (by using the same arguments as in \cite{Garey}), all the vertices $C_j$ are dominated by the vertices of the triangles, and the vertex $z$ is dominated by $z'$. As a consequence, $z$ is not a dominating forced vertex of $G_F'$.
\end{proof}

Having in mind the result above, we can deduce the following.

\begin{theorem}
\label{th:complex}
The DFV problem is co-NP-hard.
\end{theorem}

\begin{proof}
First notice that given a vertex $v$ of a graph $G$, it cannot be in general checked on whether $v$ belong to every dominating set of $G$ of cardinality $\gamma(G)$, since it is indeed a difficult problem to find the domination number of $G$. Thus, the DFV problem might not be in NP in general. Now, according to Lemma \ref{lem:complex}, it is satisfied that the vertex $z$ of the graph $G_F'$ (of the construction presented above) is a dominating forced vertex if and only if the corresponding formula $\mathcal{F}$ is not satisfiable. As a consequence, it is deduced that the DFV problem is equivalent to the complement of an NP-complete problem (3-SAT). Therefore, the deduction is completed.
\end{proof}

As a consequence of the result above, if one wants to detect the crucial points regarding a transportation sharing system as described in the previous section, then we simply cannot succeed in this task, since it is clearly ``computationally'' impossible to do it. In this sense, the remaining options are that of giving some guesses on such key vertices. To this end, one previous possible action might be giving some bounds on the number of such vertices of the transportation sharing system in question. Next two sections  deals with this research direction.

\section{On graphs having no dominating forced vertices}\label{sec:zero-DF}

The non existence of dominating forced vertices in a graph is also clearly of interest since this represents a situation in which there are no such ``crucial'' vertices in a network that might represent key points in a transportation system, when considering applications. Also, it is of theoretical interest based on the fact that this might represent a situation in which dominating sets are somehow scattered along the graph. Some trivial  examples of this situation are complete graphs, cycles, and tori, which are vertex transitive graphs. However, non vertex transitive graph could also have no dominating forced vertices.

Finding all the graphs without dominating forced vertices is clearly the same as characterizing all the graphs $G$ achieving the equality in the lower bound of \eqref{eq:triv-bounds}. Such problem seems to be indeed very challenging due to the wide structure that such graphs can have. However, if the problem is reduced to the case of trees, then it becomes easier, as already known from works like \cite{Bouquet,Mynhardt,Ziemann}, where efficient algorithms have been described for finding all dominating forced vertices, among other results. However, specific structural properties for the families of trees having no dominating forced vertices are not known. In this sense, we next give a constructive structural characterization for them. To this end, we describe two operations on trees that built bigger trees.


\medskip
\noindent Operation $\cO_1$: a tree $T$ is obtained from a tree $T'$ by adding a new edge $vx$ and the edge $uv$, where $u\in V(T')$ and there exists a private neighbor $w_S$ with respect to every (if any) $\gamma(T')$-set $S$ that contains $u$. In such a case we use the notation $T=\cO_1(T')$. See the upper diagram of Figure~\ref{f:cO1}.\medskip


\noindent Operation $\cO_2$: a tree $T$ is obtained from a tree $T'$ by adding a new path $vxy$ and the edge $uv$, where $u\in V(T')$ belongs to at least one $\gamma(T')$-set. Similarly as before,we use $T=\cO_2(T')$. See the lower diagram of Figure~\ref{f:cO1}.\medskip

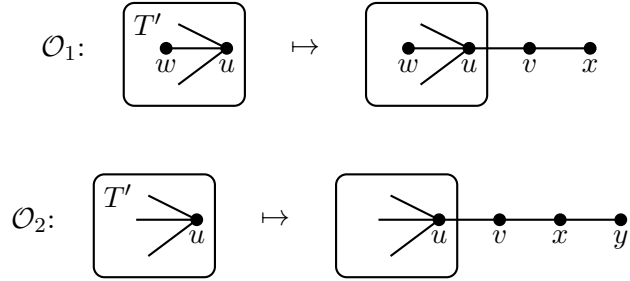
\begin{figure}[htb]

\begin{center}
\begin{tikzpicture}[scale=.8,style=thick,x=1cm,y=1cm]
\def\vr{2.5pt} 
\path (1.7,0.75) coordinate (u);
\path (0.7,.75) coordinate (w);
\path (0.9,1.15) coordinate (w1);
\path (0.9,0.15) coordinate (w2);
%
\draw (u) [fill=black] circle (\vr);
\draw (w) [fill=black] circle (\vr);
\draw (u) -- (w);
\draw (w1)--(u)--(w2);
\draw [rounded corners] (0,-.2) rectangle (2,1.5);
\draw (.4,1.15) node {$T'$};
\draw (-1,.75) node {$\cO_1$:};
\draw[anchor = north] (u) node {$u$};
\draw[anchor = north] (w) node {$w$};
\draw (3,0.75) node {$\mapsto$};
\path (5.7,.75) coordinate (u);
\path (4.7,0.75) coordinate (w);
\path (6.7,.75) coordinate (v);
\path (7.7,0.75) coordinate (x);
\path (4.9,1.15) coordinate (w1);
\path (4.9,0.15) coordinate (w2);
%
\draw (w)--(u)--(v)--(x);
\draw (w1)--(u)--(w2);
\draw (u) [fill=black] circle (\vr);
\draw (w) [fill=black] circle (\vr);
\draw (v) [fill=black] circle (\vr);
\draw (x) [fill=black] circle (\vr);
\draw [rounded corners] (4,-.2) rectangle (6,1.5);
\draw[anchor = north] (u) node {$u$};
\draw[anchor = north] (w) node {$w$};
\draw[anchor = north] (v) node {$v$};
\draw[anchor = north] (x) node {$x$};
\end{tikzpicture}
\end{center}

\vskip 0.2 cm

\begin{center}
\begin{tikzpicture}[scale=.8,style=thick,x=1cm,y=1cm]
\def\vr{2.5pt} 
\path (1.7,0.75) coordinate (u);
\path (0.7,.75) coordinate (w);
\path (0.9,1.15) coordinate (w1);
\path (0.9,0.15) coordinate (w2);
%
\draw (u) [fill=black] circle (\vr);
\draw (u) -- (w);
\draw (w1)--(u)--(w2);
\draw [rounded corners] (0,-.2) rectangle (2,1.5);
\draw (.4,1.15) node {$T'$};
\draw (-1,.75) node {$\cO_2$:};
\draw[anchor = north] (u) node {$u$};
\draw (3,0.75) node {$\mapsto$};
\path (5.7,.75) coordinate (u);
\path (4.7,0.75) coordinate (w);
\path (4.9,1.15) coordinate (w1);
\path (4.9,0.15) coordinate (w2);
\path (6.7,.75) coordinate (v);
\path (7.7,0.75) coordinate (x);
\path (8.7,0.75) coordinate (y);
%
\draw (w)--(u)--(v)--(x)--(y);
\draw (w1)--(u)--(w2);
\draw (u) [fill=black] circle (\vr);
\draw (y) [fill=black] circle (\vr);
\draw (v) [fill=black] circle (\vr);
\draw (x) [fill=black] circle (\vr);
\draw [rounded corners] (4,-.2) rectangle (6,1.5);
\draw[anchor = north] (u) node {$u$};
\draw[anchor = north] (y) node {$y$};
\draw[anchor = north] (v) node {$v$};
\draw[anchor = north] (x) node {$x$};
\end{tikzpicture}
\end{center}
\vskip -0.5cm

\caption{The operations $\cO_1$ and $\cO_2$.} \label{f:cO1}
\end{figure}

A tree $T$ belongs to the class of trees $\cT$ if $T$ can be obtained from $K_2$ by a sequence of operations $\cO_1$ and $\cO_2$. Notice that $\cO_1(K_2)=P_4$ and $\cO_2(K_2)=P_5$. Also, we cannot perform $\cO_1$ on a leaf $u$ of $P_4$ because $u$ belongs to a $\gamma(P_4)$-set and there exists a $\gamma(P_4)$-set where $u$ has no private neighbor. Vertex $u$ from $T'$ in operation $\cO_1$ can also be outside of any $\gamma(T')$-set. The smallest example for this is the middle vertex of $P_5$ that belongs to no $\gamma(P_5)$-set but it can be used as $u$ in consecutive use of $\cO_1$ to obtain a subdivided star in $\cT$.

\begin{theorem} \label{main}
Let $T$ be a tree. Then $\df(T)=0$ if and only if $T\in \cT$.
\end{theorem}

\begin{proof}
Assume first that $T\in \cT$. We will show that $\df(T)=0$ by induction on the length $k$ of the sequence of operations $\cO_1$ and $\cO_2$. Clearly, $\df(K_2)=0$ and the basis holds. Let $T'$ be a tree obtained from $K_2$ by some sequence of operations $\cO_1$ and $\cO_2$ of length $k$. By induction hypothesis we have $\df(T')=0$. Therefore, for any $z\in V(T')$, there exists a $\gamma(T')$-set $S_z$ such that $z\notin S_z$.

Assume first that $T=\cO_1(T')$, where the new edge $vx$ is added and the edge $uv$ is added to $u\in V(T')$. According to the operation, notice that in every (if any) $\gamma(T')$-set $S$ that contains $u$ there exists a private neighbor $w_S$ of $u$ with respect to $S$. Hence, since $w_S$ is a private neighbor of $u$ with respect to $S$, we observe that the set $S\setminus\{u\}\cup\{v\}$ is not a dominating set of $T$. Also, if $u$ does not belong to a $\gamma(T')$-set $S'$, in order to dominate $x,v$, at least one of them must be added to $S'$ to dominate $T$. In any case, we deduce that $\gamma(T)=\gamma(T')+1$. Clearly, $A=S_z\cup\{v\}$ and $B=S_z\cup\{x\}$ are $\gamma(T)$-sets and $z$ does not belong to them. So, any vertex $z\in V(T')$ is not a dominating forced vertex of $T$. Also, $v$ and $x$ are not dominating forced vertices of $T$ because $v\notin B$ and $x\notin A$ (for any $z$). Thus, $\df(T)=0$ and we are done with operation $\cO_1$.

For operation $\cO_2$, let $T=\cO_2(T')$, where $vxy$ is a new path together with the edge $uv$ and $u\in V(T')$ belong to a $\gamma(T')$-set $D$. In this situation, we readily see that $A=D\cup\{y\}$ is a $\gamma(T)$-set and $B_z=S_z\cup\{x\}$ is a $\gamma(T)$-set. Hence, $\gamma(T)=\gamma(T')+1$. Moreover, $z,v,y\notin B_z$ for any $z\in V(T')$ and $x\notin A$. Again $T$ contains no dominating forced vertices and we have $\df(T)=0$. Both operations $\cO_1$ and $\cO_2$ preserve the property of no dominating forced vertices and we are done with this direction.

Conversely assume that $\df(T)=0$ holds for a tree $T$ of order $n$. We will show that $T\in \cT$ by induction on $n$. Clearly, $n\geq 2$. If $n=2$, then $T=K_2$ and $T\in \cT$. There is no three $T$ with $n=3$ and $\df(T)=0$. If $\df(T)=0$ and $n=4$, then $T=P_4$ and $T=\cO_1(K_2)\in \cT$. If $\df(T)=0$ and $n=5$, then $T=P_5$ and $T=\cO_3(K_2)\in \cT$. Let now $n>5$.

Our strategy in the rest of the proof is to reduce the tree $T$ to a tree $T'$ with $\df(T')=0$ of order less than~$n$, apply the inductive hypothesis to $T'$, that is $T'\in \cT$, and then reconstruct $T$ from $T'$ by applying one of the operations $\cO_1$ or $\cO_2$, to show that $T \in \cT$. We state this formally, since we will frequently use the following statement.

\begin{statement}\label{stat}
If $\df(T')=0$ for a tree $T'$ of order less than~$n$ and $T$ can be constructed from $T'$ by applying one of the operations $\cO_1$ or $\cO_2$, then $T \in \cT$.
\end{statement}

Let $a$ and $a'$ be diametrical vertices of $T$, that is, vertices at a largest possible distance in $T$, and denote the path between $a$ and $a'$ by $P$. Clearly, $a$ and $a'$ are leaves and let further $b$ be the neighbor of $a$. If ${\rm deg}(b)>2$, then by the maximality of $P$, the vertex $b$ must be a strong support, and so, it is a dominating forced vertex, a contradiction with $\df(T)=0$. So, ${\rm deg}(b)=2$ and let $c\neq a$ be the other neighbor of $b$ on $P$.

Suppose first that ${\rm deg}(c)\geq 3$. If $c$ is a support, then it is a weak support because $\df(T)=0$. Assume that $w$ is the leaf adjacent to $c$. If ${\rm deg}(c)=3$, then let $d$ be the missing neighbor of $c$. If $d=a'$, then $c$ is a dominating forced vertex of $T$, a contradiction. Hence there exists a neighbor $e\neq c$ of $d$ on $P$. If ${\rm deg}(c)\geq 4$, then let $b_i\neq w$ be a neighbor of $b$ that is not on $P$ for $i\in\{1,\dots,{\rm deg}(c)-3\}$. In this case $b_i$ has exactly one leaf neighbor $a_i$ for every $i\in\{1,\dots,{\rm deg}(c)-3\}$ because $\df(T)=0$. So, diameter of $T$ is at least four and there exists a neighbor $d\neq b$ of $c$ on $P$ and a neighbor $e\neq c$ of $d$ on $P$.  In both cases let $T'=T-\{a,b\}$. Notice that $c$ has a private neighbor $w$ in every $\gamma(T')$-set that contains $c$. Therefore $T=\cO_1(T')$ for $x=a,v=b$ and $u=c$ and $T\in\cT$ by Statement \ref{stat}.

Next, we have no leaf adjacent to $c$. Since ${\rm deg}(c)\geq 3$, there exist some neighbors $b_i$, $i\in\{1,\dots,{\rm deg}(c)-2\}$, of $b$ that are not on $P$. Again, $b_i$ has exactly one leaf neighbor $a_i$ for every $i\in\{1,\dots,{\rm deg}(c)-2\}$ because $c$ has no leaf and because $\df(T)=0$. As before, the diameter of $T$ is at least four and there exists a neighbor $d\neq b$ of $c$ on $P$ and a neighbor $e\neq c$ of $d$ on $P$. Let $S$ be a $\gamma(T)$-set that contains $c$ if it exists. Vertices $b,b_1,\dots,b_{{\rm deg}(c)-2}$ are not private neighbors of $c$. If also $d$ is not a private neighbor of $c$, then $(S-\{c,a\})\cup\{b\}$ is a dominating set of $T$ of cardinality smaller than $S$, a contradiction. So, $c$ has a private neighbor $d$ with respect to any $\gamma(T)$-set that contains $c$ (if any). For $T'=T-\{a,b\}$ we have  $T=\cO_1(T')$ for $x=a,v=b$ and $u=c$ and $T\in\cT$ by Statement \ref{stat}.

We are left with possibility that ${\rm deg}(c)= 2$ and again let $d\neq b$ be a neighbor of $c$ on $P$. If $d=a'$, then $T=P_4$ and $T=\cO_1(K_2)$. So, $T\in\cT$ by Statement \ref{stat} since $K_2\in\cT$. Otherwise, $d\neq a'$ and let $e\neq c$ a neighbor of $d$ on $P$. Let $T'=T-\{a,b,c\}$. If $d$ is in no $\gamma(T')$-set $S$, then we have $\gamma(T)=\gamma(T')+1$ and $S\cup\{b\}$ is a $\gamma(T)$-set. Hence, $b$ is a dominating forced vertex, a contradiction. So, $d$ belongs to a $\gamma(T')$-set and we have $T=\cO_2(T')$ for $y=a,x=b,v=c$ and $u=d$. By Statement \ref{stat} we have $T\in\cT$.
\end{proof}

\section{An upper bound on the value of $\df(G)$}\label{sec:upper-bound}

It is nowadays a folklore result that the domination number of an isolated free graph $G$ of order $n$ is at most $n/2$. Based on the inequality \eqref{eq:triv-bounds}, this leads to $\df(G)\le n/2$. However, as we next show this can be significantly reduced.

\begin{theorem}\label{upper}
If $G$ is a graph of order $n$, then $\df(G)\le n/3$ and this bound is tight.
\end{theorem}

\begin{proof}
If $G$ has no dominating forced vertices, then clearly $\df(G)\le n/3$. Hence, we may assume that $G$ has at least one dominating forced vertex $x$. Let $D\subset V(G)$ be a $\gamma(G)$-set. Clearly $x\in D$. If $pn(x,D)=\emptyset$, then $(D-\{x\})\cup\{y\}$, $y\in N(x)$, is a $\gamma(G)$-set without $x$, a contradiction. So, $|pn(x,D)|\geq 1$.

Suppose $|pn(x,D)|=1$ and let $x'\in pn(x,D)$. Hence, we readily observe that the new set $D'=D\setminus \{x\}\cup \{x'\}$ is also a dominating set of $G$ (indeed a $\gamma(G)$-set), and this is not possible since $x$ is a dominating forced vertex. Thus,
\begin{equation}
\label{eq:pn}
|pn(x,D)|\ge 2.
\end{equation}

Assume now that $D'\subseteq D$ is the set of dominating forced vertices of $G$, which means that $\df(G)=|D'|$. Let $S=\bigcup_{v\in D'} pn(v,D)$ and let $Q=V(G)\setminus (D\cup S)$. Notice that $|S|=\left|\bigcup_{v\in D'} pn(v,D) \right|\ge 2|D'|=2\cdot\df(G)$. On the other hand, it follows that
\begin{equation}
\label{eq:bound-n-dfv}
n=|D'|+|D\setminus D'|+|S|+|Q|\ge \df(G)+|D\setminus D'|+2\cdot\df(G)+|Q|\ge 3\cdot\df(G),
\end{equation}
which leads to the desired bound.

To see the tightness of the bound, we consider a corona graph $H\odot N_2$ for any graph $H$ and the edgeless graph $N_2$. It can be then easily checked that $\df(H\odot N_2)=\gamma(H\odot N_2)=|V(H)|$, since $H\odot N_2$ has a unique minimum dominating set of cardinality $|V(H)|$ formed by the set of vertices of $H$.
\end{proof}

Next we give a structural characterization of all graphs that are sharp for the upper bound above. In the proof we strongly rely on the proof of Theorem \ref{upper}, in particular on (\ref{eq:pn}) and (\ref{eq:bound-n-dfv}).

\begin{theorem}\label{char}
Let $G$ be a graph of order $n$.  Then $\df(G)=n/3$ if and only if $G$ has the unique $\gamma(G)$-set $D$ such that $G[D]$ contains components $G_1,\dots,G_k$ where $G[N[V(G_i)]]=G_i\odot N_2$ for every $i\in\{1,\dots,k\}$ and $N[V(G_1)],\dots,N[V(G_k)]$ is a partition of $V(G)$.
\end{theorem}

\begin{proof}
Let first $G$ be a graph with $\df(G)=n/3$. As in the proof of Theorem \ref{upper} let $D$ be a $\gamma(G)$-set, $D'\subseteq D$ be the set of all dominating forced vertices, $S=\bigcup_{v\in D'} pn(v,D)$ and let $Q=V(G)\setminus (D\cup S)$. Since $\df(G)=n/3$, we have equalities in (\ref{eq:bound-n-dfv}) and the second inequality gives $Q=D\setminus D'=\emptyset$ and now the first inequality yields $|D'|=|D|=\df(G)$ and $|S|=2\df(G)$. So, all vertices of a $\gamma(G)$-set $D$ are dominating forced vertices and $D$ is the unique $\gamma(G)$-set because $D\setminus D'=\emptyset$. We have further $|pn(x,D)|=2$ for every $x\in D$ by (\ref{eq:pn}) since $|S|=2\df(G)$ and let $pn(x,D)=\{x_1,x_2\}$ for every $x\in D$. If $x_1x_2\in E(G)$, then $(D-\{x\})\cup\{x_1\}$ is a $\gamma(G)$-set, a contradiction with uniqueness of $D$. So, $G[N[V(G_i)]]=G_i\odot N_2$ for every $i\in\{1,\dots,k\}$. Finally, because $Q=\emptyset$, vertices of $N[V(G_1)],\dots,N[V(G_k)]$ partition $V(G)$.

For the other direction let $G$ be a graph with the unique $\gamma(G)$-set $D$ such that $G[D]$ contains components $G_1,\dots,G_k$ where $G[N[V(G_i)]]=G_i\odot N_2$ for every $i\in\{1,\dots,k\}$ and $N[V(G_1)],\dots,N[V(G_k)]$ is a partition of $V(G)$. As $D$ is the unique $\gamma(G)$-set, we have $\df(G)=|D|$. Let $G_1,\dots,G_k$ be components of $G[D]$. Since $G[N[V(G_i)]]=G_i\odot N_2$ for every $i\in\{1,\dots,k\}$ and $N[V(G_1)],\dots,N[V(G_k)]$ is a partition of $V(G)$, we get $\df(G)=|D|=n/3$.
\end{proof}

\section{The study cases}

Once presented our theoretical results on the dominating forced vertices of graphs, we focus in this section into applying our arguments to detect the ``best'' location points for a hypothetical transportation sharing systems that could be developed on the two main cities of the ``Campo de Gibraltar'' area in Spain. That is, the cities of Algeciras and La L\'inea de la Concepci\'on.

\subsection{The methodology}

Assume we are given with an urban network where we would be willing to develop a transportation sharing system, in which we are required to detect key points in the sense of our settings. Hence, the following process needs to be achieved.
\begin{description}
    \item[Step 1] Design the network (a graph $G$) modeling the urban area in question.
    \item[Step 2] Develop some procedures that will extract all the possible minimum dominating sets of $G$, that is the collection $\mathcal{C}=\{D_1,\dots,D_r\}$ where each $D_i\in \mathcal{C}$ is a $\gamma(G)$-set.
    \item[Step 3] Find the intersection of all the sets $D_i\in \mathcal{C}$.
\end{description}

Unfortunately, based on the difficulty of the DomSet and DFV problems, one cannot proceed with such a process. In consequence, one needs to find approximation techniques to manage the situation. In this sense, the two last steps from the three above ones can be approached in the following way. Regarding Step 2, an exact approach to determine all dominating forced vertices is computationally prohibitive, especially for large urban networks. Thus, a heuristic approach is a reasonable choice as it is able to approximate the best solution at a reasonable cost. In \cite{Casado}, a metaheuristic algorithm was proposed to find such solutions using the Iterated Greedy (IG) framework presented in Algorithm 1 from \cite{Casado}. In addition, the solutions provided by such metaheuristic algorithm are not necessarily the best ones. Thus, regarding Step 3, one can approach by considering the intersection of ``most of the best'' solutions that such algorithm will provide, even so, such solutions will not have the same cardinality. That is, by making use of their methodology, we establish a further framework to find vertices that appear in (hopefully) several minimum dominating sets of a graph.

The IG metaheuristic is a stochastic search method designed to navigate the complex landscapes of NP-hard optimization problems. Although simple greedy heuristics are computationally efficient, they are prone to becoming trapped in local optima. The IG framework overcomes this by using an iterative process that repeatedly transforms a complete solution into a partial one and subsequently reconstructs it. This cycle allows the algorithm to ``jump'' out of local minima by strategically perturbing the current best candidate while retaining its most effective structural components.

The authors in \cite{Casado} addressed the problem of finding the minimum dominating set in a graph $G$, which is referred to as the Minimum Dominating Set Problem (MDSP), using the IG framework. We call such a framework \texttt{IG-MDSP}. We recommend that readers refer to their work for details on each step. In brief, the procedure is conducted as follows: let $0 \le \beta \le 1$ and $\Delta$ be two given parameters.
\begin{enumerate}
    \item \texttt{InitialSolution($G$)}: Generate a dominating set using a greedy algorithm and perform the ``Checking Procedure" (or \texttt{CheckP}) to remove redundant vertices, producing a minimal dominating set $D$ (recall that a dominating set is minimal if it does not contain a proper subset that is also a dominating set).
    \item \texttt{LocalImprovement($D$)}: Replace a vertex in $D$ with a vertex in $V \setminus D$ that has a certain property. Perform \texttt{CheckP} again to remove redundant vertices, resulting in a new feasible solution.
    Do it iteratively, in random order, for every vertex in $D$ and every suitable exchange candidate in $V \setminus D$ until we find the first solution that has a smaller size than $D$. We call this new solution $D_b$.
    \item Set $\delta \leftarrow 0$ and repeat Steps 4–7 while $\delta < \Delta$. If $\delta = \Delta$, then go to Step 8.
    \item \texttt{Destruction($D_b,\beta$)}: Remove a portion of vertices from $D_b$, namely $\beta \cdot |D_b|$, at random. Let the resulting set be $D_d$. This set is not dominating, so a reconstruction is needed.
    \item \texttt{Reconstruction($D_d$)}: Starting from $D_d$, do the greedy algorithm to add new vertices so that it forms a dominating set, yielding the set $D_r$.
    \item Set $D_i \leftarrow$ \texttt{LocalImprovement($D_r$)}
    \item If $|D_i| < |D_b|$, then an improved solution is found, and we set $D_b \leftarrow D_i$ and $\delta \leftarrow 0$. If not, then we set $\delta \leftarrow \delta + 1$.
    \item Return $D_b$.
\end{enumerate}

As explained before, to obtain the dominating forced vertices of a graph, it is necessary to get all minimum dominating sets and look for vertices that occur in every such set. However, realistically speaking, an urban network has a relatively large size, making it difficult to find all minimum dominating sets. Thus, it is fair enough to find the approximate vertices in the sense that they appear frequently in many minimum dominating set approximations. Taking this into account, we use a simple two-step framework to obtain such vertices. First, we perform the \texttt{IG-MDSP} several times, preferably a large number of iterations, to obtain a collection of minimum dominating set approximations (which are minimal dominating sets). Then, we compile the vertices that are included in most, hopefully all, of these sets.

All computational experiments in this study were conducted on a computer equipped with an AMD Ryzen 5 7520U (4 cores, 2.80 GHz) and 16 GB of RAM. Algorithms were implemented in Python 3.12 (64-bit) using NetworkX 3.3 for graph manipulation. All implementations of \texttt{IG-MDSP} in this study use the parameters $\beta = 0.1$ and $\Delta = 100$.

\subsection{Validation results}

In order to first see that our arguments run appropriately, this subsection is focused on developing some implementations and experiments while using some graphs with a known value for the domination number and indeed for the set of their dominating forced vertices. To this end, we consider the case of the \textit{strong grid}, i.e., the strong product of two paths. An $m \times n$ strong grid, also known as the king's graph, is the graph where a vertex represents a square on an $m \times n$ chessboard and each edge represents a legal move by a king \cite{Weisstein}. For our analysis purposes, we observe the strong grid $P_{3n} \boxtimes P_{3n}$, where $V(P_{3n} \boxtimes P_{3n}) = \{(i,j) : 1 \le i,j \le 3n\}$. One can readily verify that the set $D = \{(3i-1,3j-1) : 1 \le i,j \le n\}$ forms the unique minimum dominating set of size $n^2$ of such graphs. Thus, all vertices in $D$ are dominating forced.

Table \ref{tab: strong grids} summarizes the 2,000 attempts of the \texttt{IG-MDSP} algorithm on strong grids $P_{3n} \boxtimes P_{3n}$ for small $n$, providing the frequencies of sets according to their cardinalities.

\begin{table}[htbp]
    \centering
    \caption{Size distribution of 2,000 minimal dominating sets produced by the \texttt{IG-MDSP} algorithm on strong grids $P_{3n}\boxtimes P_{3n}$.}
    \label{tab: strong grids}

    \begin{subtable}{.45\textwidth}
        \centering
        \caption{$n=2$}
        \begin{tabular}{lc}
            \toprule
            Set size & 4 \\
            \midrule
            Frequency & 2,000 \\
            \bottomrule
        \end{tabular}
    \end{subtable}
    \hfill
    \begin{subtable}{.45\textwidth}
        \centering
        \caption{$n=3$}
        \begin{tabular}{lcccc}
            \toprule
            Set size & 9 & 10 & 11 & 12 \\
            \midrule
            Frequency & 1,927 & 72 & 0 & 1 \\
            \bottomrule
        \end{tabular}
    \end{subtable}

    \vspace{1em}

    \begin{subtable}{.45\textwidth}
        \centering
        \caption{$n=4$}
        \begin{tabular}{lccccc}
            \toprule
            Set size & 16 & 17 & 18 & 19 & 20 \\
            \midrule
            Frequency & 1,167 & 657 & 152 & 17 & 7 \\
            \bottomrule
        \end{tabular}
    \end{subtable}

    \vspace{1em} 

    \begin{subtable}{\textwidth}
        \centering
        \caption{$n=5$}
        \small 
        \begin{tabular}{lcccccccccc}
            \toprule
            Set size & 25 & 26 & 27 & 28 & 29 & 30 & 31 & 32 & 33 & 34\\
            \midrule
            Frequency & 39 & 202 & 390 & 444 & 296 & 355 & 201 & 59 & 13 & 1\\
            \bottomrule
        \end{tabular}
    \end{subtable}

    \vspace{1em}

    \begin{subtable}{\textwidth}
        \centering
        \caption{$n=6$}
        \small
        \begin{tabular}{lccccccccccc}
            \toprule
            Set size & 38 & 39 & 40 & 41 & 42 & 43 & 44 & 45 & 46 & 47 & 48 \\
            \midrule
            Frequency & 1 & 9 & 37 & 138 & 323 & 443 & 467 & 330 & 177 & 60 & 15 \\
            \bottomrule
        \end{tabular}
    \end{subtable}

    \vspace{1em}

    \begin{subtable}{\textwidth}
        \centering
        \caption{$n=7$}
        \small
        \begin{tabular}{lcccccccccccccc}
            \toprule
            Set size & 55 & 56 & 57 & 58 & 59 & 60 & 61 & 62 & 63 & 64 & 65 & 66 & 67 & 68\\
            \midrule
            Frequency & 2 & 3 & 15 & 55 & 112 & 340 & 487 & 483 & 325 & 122 & 39 & 10 & 5 & 2\\
            \bottomrule
        \end{tabular}
    \end{subtable}
\end{table}

The results for the strong grids $P_{3n} \boxtimes P_{3n}$ and the larger test instances demonstrate a clear trend in the performance of the \texttt{IG-MDSP} algorithm. For smaller instances ($2 \le n \le 4$), the algorithm consistently identifies the minimum dominating set size in the vast majority of the 2,000 iterations. For example, in $n=3$ instance, the optimal size of 9 was reached in $96.35\%$ of the tests. However, as the graph complexity increases ($5 \le n \le 7$), we observe a characteristic shift:
\begin{itemize}
    \item Distribution spread: the distribution of found set sizes becomes broader, following a near-normal distribution centered above the best-known value.
    \item Low best-case frequency proportion: for $n=5$, the best-found size, which is 25, was achieved in only $1.95\%$ of the iterations.
\end{itemize}
This behavior is expected for co-NP-hard problems. From an urban planning perspective, these results suggest that while finding the absolute minimum $\gamma(G)$ becomes computationally expensive as the network grows, the algorithm remains highly stable. By taking those vertices that appear in all minimal dominating sets produced by the algorithm, we can reasonably suggest that they act like a dominating forced vertex and have a high probability of being the ``non-negotiable'' locations for station placement.

\subsection{The particular study case for the city of Algeciras}

We are now ready to apply the mentioned methodology to study the dominating forced vertices approximation in a particular urban network, namely the city of Algeciras.

First, we extract the Algeciras road network using OSMnx \cite{Boeing}, a Python package designed to retrieve, model, and analyze street networks from OpenStreetMap (\texttt{https://www.openstreetmap.org}). OpenStreetMap is an open-access, crowd-sourced geospatial dataset maintained under the Open Database License. We refer the reader to Boeing \cite{Boeing} and the official documentation (\texttt{osmnx.readthedocs.io}) for a complete introduction to the package. Using OSMnx, we download the street network data of Algeciras directly from the OpenStreetMap database and construct a corresponding graph representation as a \texttt{NetworkX} graph object.


In our study, we focus on nodes within a radius of 3 kilometers centered at Plaza Alta (the main and central square of the city). Thus, the original map (generated on August 12th, 2026) is shown in Figure \ref{fig: algeciras raw}, where Plaza Alta is indicated by the large blue node. It is a connected graph consisting of 5,819 vertices and 13,852 edges.
With the map at hand, which includes all types of roads, we continue to decide which types of roads to consider. For our purposes, we consider only urban streets where scooters, the focus of our study on shared-transportation, are allowed to enter. This includes only main urban streets (which usually have bike lanes or are permitted), neighborhood streets, small pedestrian paths, and dedicated bike paths. Thus, the filtered map, having 1,962 vertices and 4,257 edges, is shown in Figure \ref{fig: algeciras filtered}. This version of the map may still include self-loops and multiple edges. To avoid errors in counting vertex degrees, which could harm the greedy steps of the algorithm, these loops and multiple edges are removed from our graph without deleting their corresponding vertices and edges. The resulting graph, which has 1,962 vertices and 2,671 edges, is ready to be analyzed.

We then performed the \texttt{IG-MDSP} algorithm on this graph for 2,000 iterations, resulting in the same number of minimal dominating sets. The size distribution is summarized in Table \ref{tab: algeciras 3km}. Simple checking confirms that all sets produced are unique. It is apparent that the size distribution of the sets forms a normal-like shape, which is reasonable given the randomness of the greedy steps. From the data, we retrieved all vertices with the highest frequency of appearance across all the sets produced. Our experiment resulted in 38 such vertices out of 1,962 available (about 1.94\%). These selected vertices, which are the candidates for placing the stations, are shown in Figure \ref{fig: algeciras result} as the large red nodes, as well as, in Figure \ref{fig: algeciras final map}, over the real street view map of Algeciras.

\begin{table}[ht]
    \centering
    \caption{Size distribution of minimal dominating sets produced by the \texttt{IG-MDSP} algorithm on Algeciras' map network.}
    \label{tab: algeciras 3km}
    \small
    \begin{tabular}{lcccccccccccccc}
        \toprule
        Set size & 661 & 662 & 663 & 664 & 665 & 666 & 667 & 668 & 669 & 670 & 671 & 672 & 673 & 674 \\
        \midrule
        Frequency & 1 & 1 & 7 & 5 & 16 & 30 & 52 & 67 & 101 & 161 & 180 & 202 & 255 & 231 \\
        \bottomrule
    \end{tabular}

    \bigskip

    \begin{tabular}{lcccccccccc}
        \toprule
        Set size & 675 & 676 & 677 & 678 & 679 & 680 & 681 & 682 & 683 & 684 \\
        \midrule
        Frequency & 196 & 170 & 114 & 71 & 81 & 25 & 21 & 7 & 4 & 2 \\
        \bottomrule
    \end{tabular}
\end{table}


\subsection{The particular study case for the city of La L\'inea de la Concepci\'on}

Having the established model for Algeciras at hand, in order to complete our study for the main cities in the area of ``Campo de Gibraltar'', we continue by performing the exact same methodology on the neighboring city of La L\'inea de la Concepci\'on.

We generate the road network of the city through OSMnx (see Figure \ref{fig: lalinea raw}), having a radius of 3 kilometers centered at the city's main square: Plaza de Toros El Arenal, presented as the large blue node. The map, generated on August 12, 2026, has 5,602 vertices and 13,386 edges. For our study purposes, we choose only the road types where scooters are permissible, resulting in a new ``filtered map"  having 2,793 vertices and 5,610 edges (see Figure \ref{fig: lalinea filtered}). This graph is directed and has multiple edges and loops. To produce the minimal dominating sets, we need to convert the graph into a simple, undirected graph by removing the aforementioned instances. The final graph, ready to be analyzed, has 2,793 vertices and 3,694 edges.

We ran the \texttt{IG-MDSP} algorithm on the final graph to produce 2,000 minimal dominating sets, taking roughly two days of computing time. All sets produced are unique, and the size distribution is presented in Table \ref{tab: lalinea 3km}. We extracted all vertices with the highest appearance frequency over all produced sets, yielding 89 such vertices out of 2,793 available (about 3.19\%). Each of these vertices has a frequency of 2,000, indicating that they appear in every minimal dominating set produced by the algorithm. The final map is presented in Figure \ref{fig: lalinea final} where these selected vertices (shown as large red nodes) indicate the positional candidates for scooter stations, as well as, in Figure \ref{fig: lalinea final map}, over the real street view map of La L\'inea de la Concepci\'on.

\begin{table}[ht]
    \centering
    \caption{Size distribution of minimal dominating sets produced by the \texttt{IG-MDSP} algorithm on La L\'inea de la Concepci\'on's map network.}
    \label{tab: lalinea 3km}
    \small
    \begin{tabular}{lcccccccccccccc}
        \toprule
        Set size & 944 & 945 & 946 & 947 & 948 & 949 & 950 & 951 & 952 & 953 & 954 & 955 & 956 & 957 \\
        \midrule
        Frequency & 1 & 1 & 0 & 0 & 1 & 8 & 6 & 14 & 21 & 47 & 76 & 129 & 150 & 178 \\
        \bottomrule
    \end{tabular}

    \bigskip

    \begin{tabular}{lccccccccccccccc}
        \toprule
        Set size & 958 & 959 & 960 & 961 & 962 & 963 & 964 & 965 & 966 & 967 & 968 & 969 & 970 & 971 & 972\\
        \midrule
        Frequency & 224 & 213 & 217 & 158 & 175 & 137 & 86 & 68 & 47 & 17 & 15 & 6 & 4 & 0 & 1 \\
        \bottomrule
    \end{tabular}
\end{table}

\subsection{Some conclusions on the empirical evaluations}

First of all, while the exact computation of dominating forced vertices is co-NP-hard, as shown in Theorem \ref{th:complex}, the IG-MDSP algorithm demonstrates that metaheuristic frequency analysis is a viable method for identifying critical infrastructure locations in (relatively) large-scale real-world networks.

On the other hand, the results of the experimental evaluations for the two cities in the Campo de Gibraltar area show the following. In Algeciras, the framework identified 38 high-frequency candidate nodes out of 1,957 available vertices (about $1.94\%$). This indicates a highly localized set of crucial connection points within the street network for placing shared mobility stations, while in La L\'inea de la Concepci\'on, the analysis yielded 89 candidate nodes out of 2,793 available vertices (about $3.19\%$). This higher count and perfect recurrence highlight a stronger structural necessity for distributed stations across its street grid layout. However, in both cities, the counts for the non-negotiable node candidates (between $1.9\%$ and $3.2\%$) are far below the theoretical maximum bound of $33.3\%$ (regarding $n(G)/3$ from Theorem \ref{upper}). This would ensure that city planners might achieve near-optimal initial coverage ($100\%$ accessibility within neighborhood hops) with a highly minimal capital expenditure. These high-frequency nodes may serve as a ``kind-of optimal Phase 1'' infrastructure blueprint for micro-mobility deployment in these two cities.

In addition, we may also conclude that in Algeciras, the planning may remain relatively more flexible because fewer nodes are strictly forced. Hence, operators might have leeway to adjust station positions based on secondary real-world constraints (such as land availability, power access, or commercial interest) without significantly degrading overall system equity. On the other hand, in La L\'inea de la Concepci\'on, the 89 identified nodes represent rigid, non-negotiable anchor points. The omission of any of these stations would significantly increase coverage costs or require disproportionately larger station fleets to maintain equitable access across adjacent neighborhoods.

\section*{Concluding remarks}

We have considered in this work a graph theoretical approach for detecting key points that might be of interest for placing non-negotiable location stations in transportation sharing systems that would be addressed to keep some equity or accessibility features. That is, we have focused into modeling a situation in which locations for such sharing systems will be placed as near as possible to every possible client in the urban area that would be considered, and have aimed to detect the most crucial points for such possible locations. In this sense, a natural model for such a situation is that of dominating sets which naturally arises in the area of graph theory. Having such setting in mind, we then focused on detecting those more important points that are somehow forced to have a location for placing the transportation devices. Such points are then understood as vertices that belong to every dominating set of the smallest possible cardinality.

In the study, we have first shown that finding such points is indeed a computationally hard problem, and have presented some heuristic to obtain possible approximated solutions. The procedures have been applied to the two cities of Campo de Gibraltar, and such important points have been detected. Moreover, some theoretical contributions on the vertices that belong to every dominating set of a graphs have been derived. As a natural continuation of our investigation, we consider the following research lines are worthy of exploring.

\begin{itemize}
    \item Improve the heuristic for finding optimal dominating sets of the graph.
    \item Develop graph theoretical studies (bounds and properties) on the dominating forced vertices of different graph classes which would deliver better knowledge about them. In particular, graphs that are planar or close to planar seem to be of interest, because the road network is often close to planar graphs.
    \item Design some polynomial algorithm to find such dominating forced vertices in particular network structures.
    \item Consider some other more realistic requirements of location stations in transportation sharing systems so that the model used in our investigations will better describe the existence of some related crucial locations.
    \item Studying the number of dominating forced vertices of random graphs seems to be an interesting task in order to test more the capabilities of the IG-MDSP algorithm.
\end{itemize}

\section*{Data and implementations availability}

All codes of the proposed algorithm and the datasets generated during the study are publicly available in the GitHub repository at \texttt{https://github.com/farhan-math/domination}. All code was implemented using Python 3.12.

\section*{Acknowledgements}

This work has been supported by the European Commission's Horizon Europe Research and Innovation programme through the Marie Sklodowska-Curie Actions Staff Exchanges (MSCA-SE) under Grant Agreement no.101182819 (COVER: (C)ombinatorial (O)ptimization for (V)ersatile Applications to (E)merging u(R)ban Problems). I. Peterin was partially supported by the Slovenian Research and Innovation Agency (research core funding No.\ P1-0297). I.\ G.\ Yero has been partially supported by the Spanish Ministry of Science and Innovation through the grant PID2022-139543OB-C41.

\section*{Conflict of interest disclosure statement}

Dorota Kuziak is an Editor of the Open Mathematics journal and was not involved in the review and decision-making process of this article.

\newpage
\section*{Appendix}

\begin{figure}[h]
    \centering
    \begin{subfigure}{0.4\textwidth}
    \includegraphics[width=1.2\textwidth]{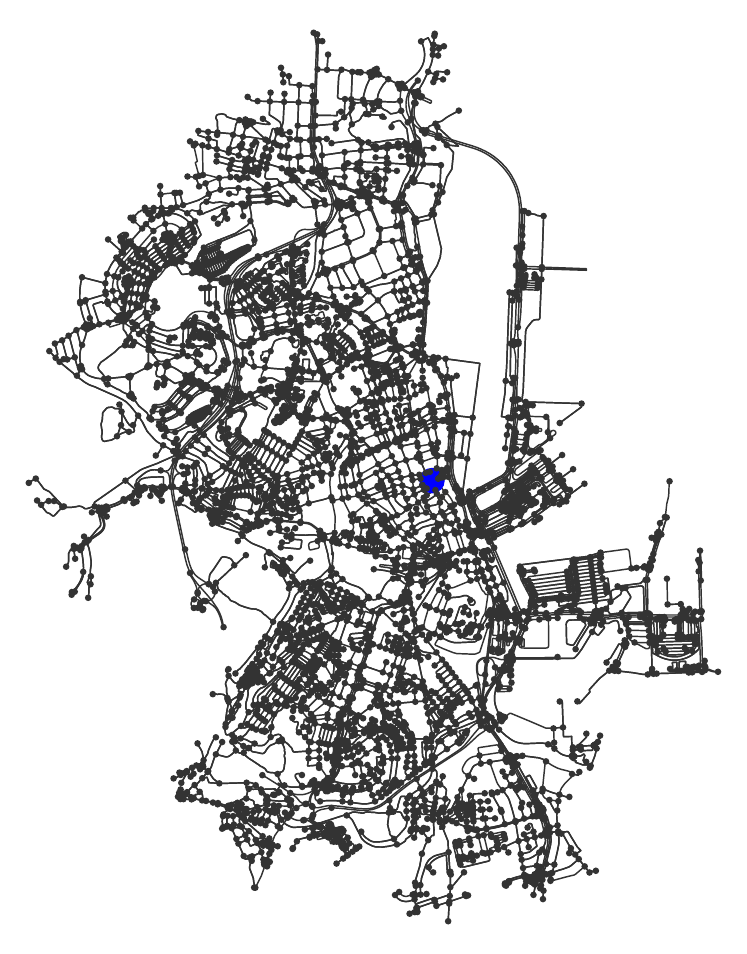}
    \caption{Original map}
    \label{fig: algeciras raw}
    \end{subfigure}
    \hfill
    \begin{subfigure}{0.4\textwidth}
    \includegraphics[width=\textwidth]{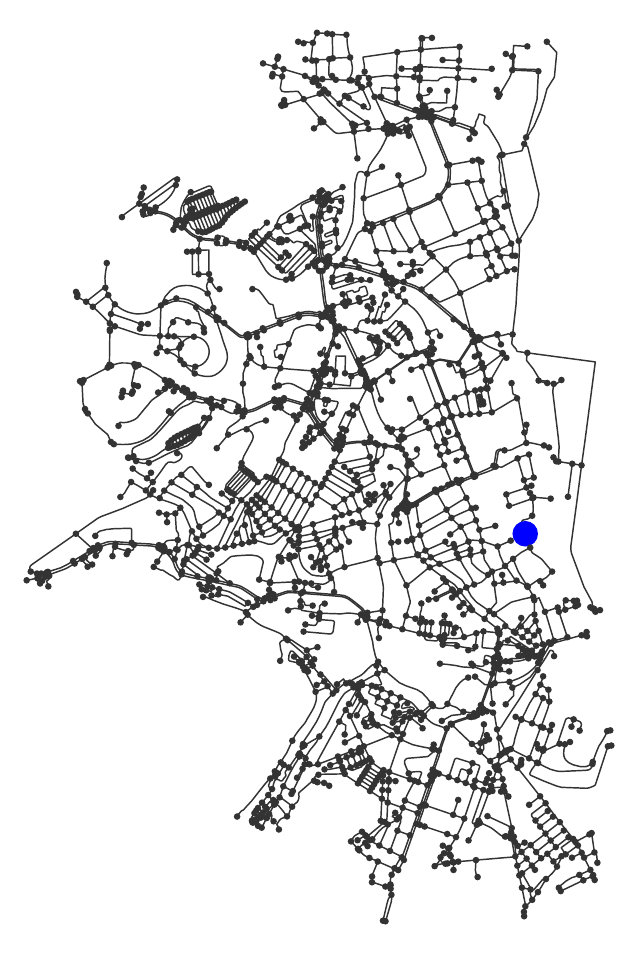}
    \caption{Filtered map}
    \label{fig: algeciras filtered}
    \end{subfigure}
    \caption{The road network map of Algeciras with radius 3 km centered at Plaza Alta.}
\end{figure}

\begin{figure}[h]
    \centering
    \includegraphics[width=0.9\linewidth]{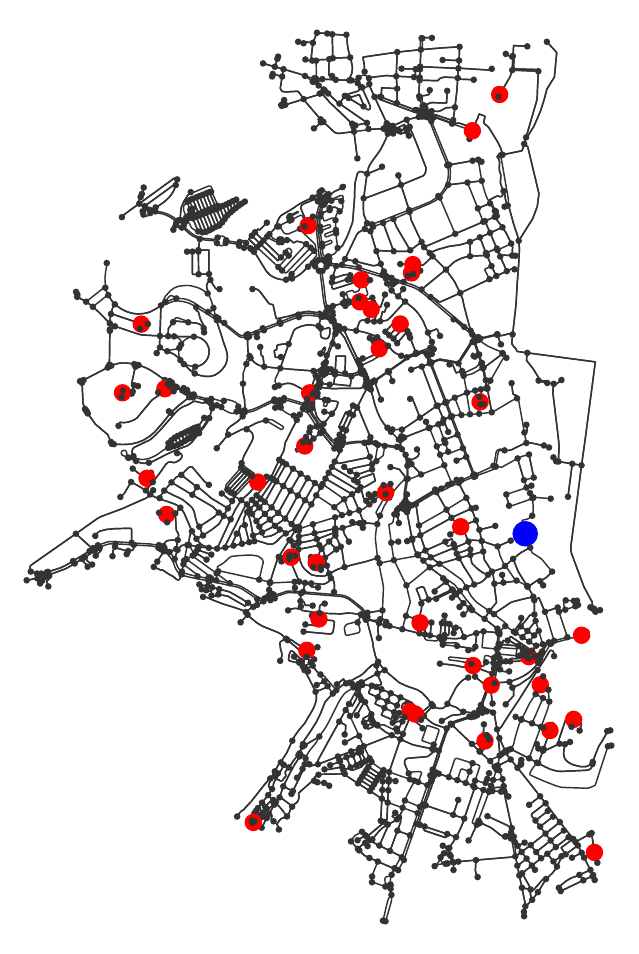}
    \caption{The filtered road network of Algeciras and nodes with highest frequency of appearance}
    \label{fig: algeciras result}
\end{figure}

\begin{figure}[h]
    \centering
    \includegraphics[width=0.85\linewidth]{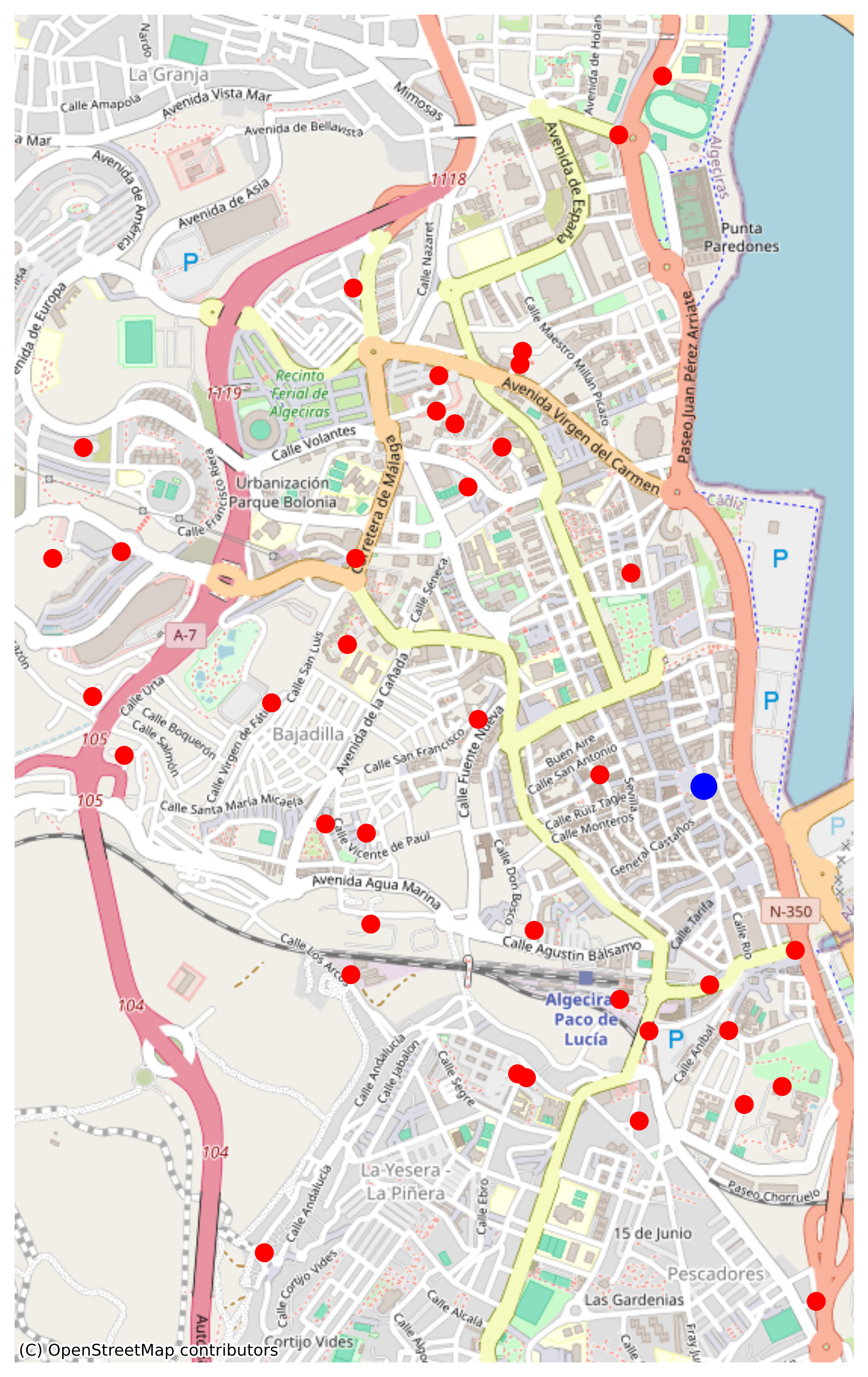}
    \caption{The map of Algeciras and nodes with highest frequency of appearance. Map data \copyright~OpenStreetMap contributors (openstreetmap.org/copyright).}
    \label{fig: algeciras final map}
\end{figure}

\begin{figure}[h]
    \centering
    \begin{subfigure}{0.5\textwidth}
    \includegraphics[width=1\textwidth]{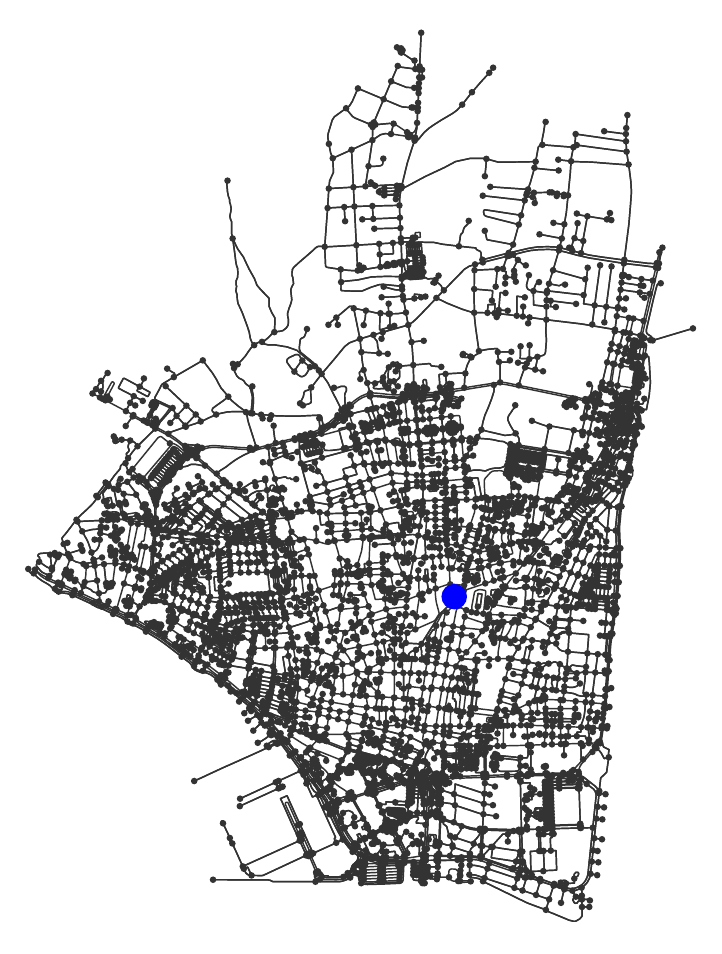}
    \caption{Original map}
    \label{fig: lalinea raw}
    \end{subfigure}
    \\
    \begin{subfigure}{0.5\textwidth}
    \includegraphics[width=1\textwidth]{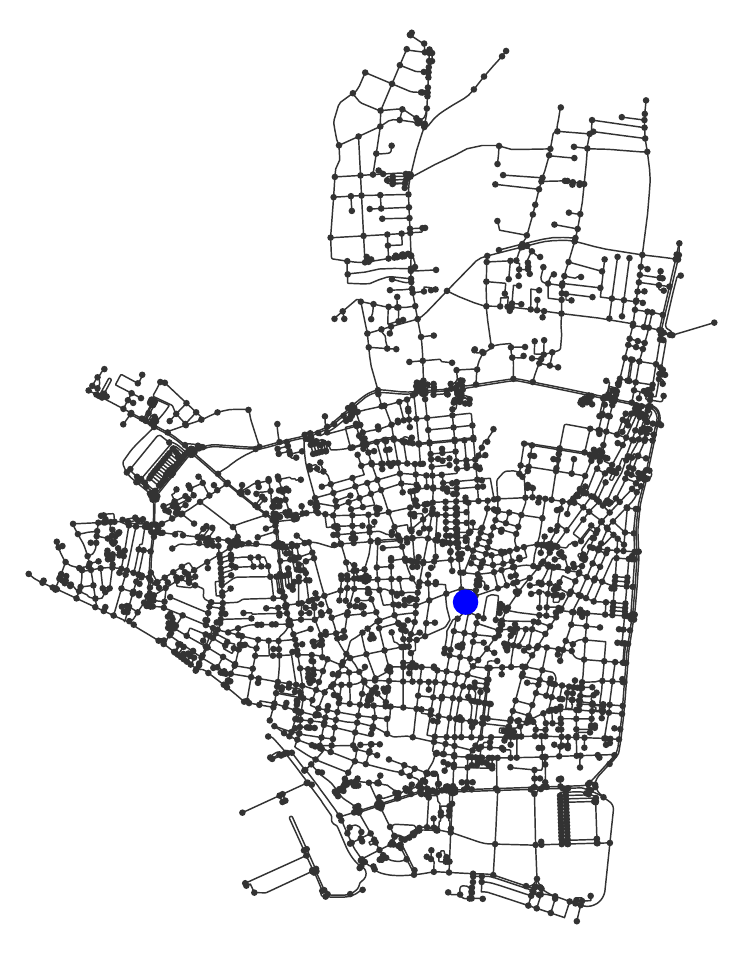}
    \caption{Filtered map}
    \label{fig: lalinea filtered}
    \end{subfigure}
    \caption{The road network map of La L\'inea de la Concepci\'on with radius 3 km centered at Plaza de Toros El Arenal.}
\end{figure}

\begin{figure}[h]
    \centering
    \includegraphics[width=0.9\linewidth]{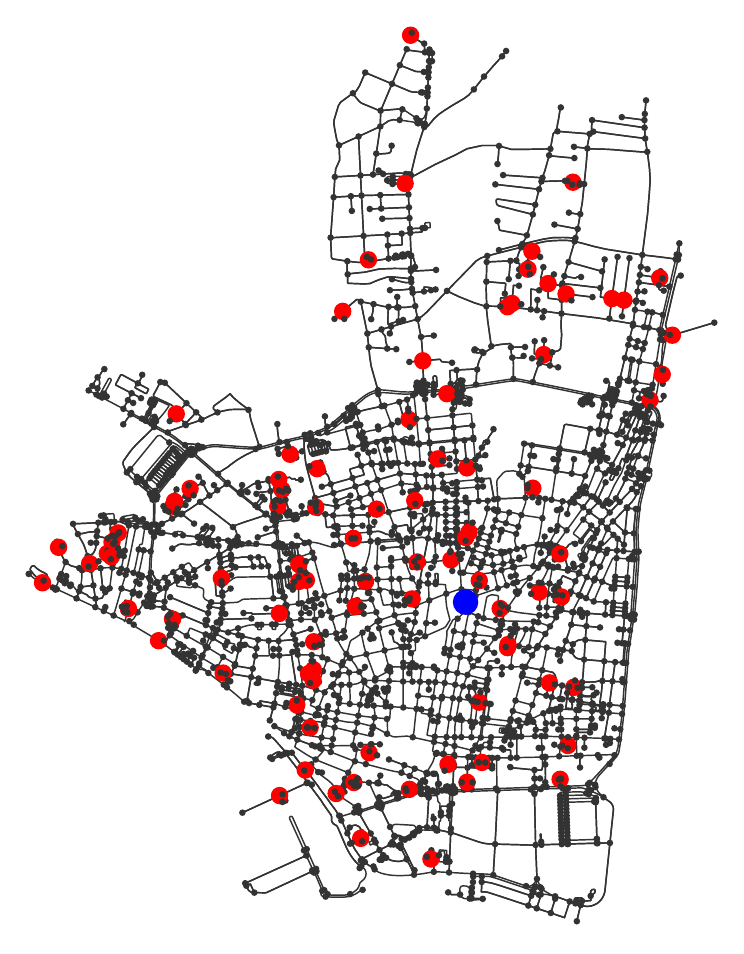}
    \caption{The road network of La L\'inea de la Concepci\'on and nodes with highest frequency of appearance}
    \label{fig: lalinea final}
\end{figure}

\begin{figure}[h]
    \centering
    \includegraphics[width=0.85\linewidth]{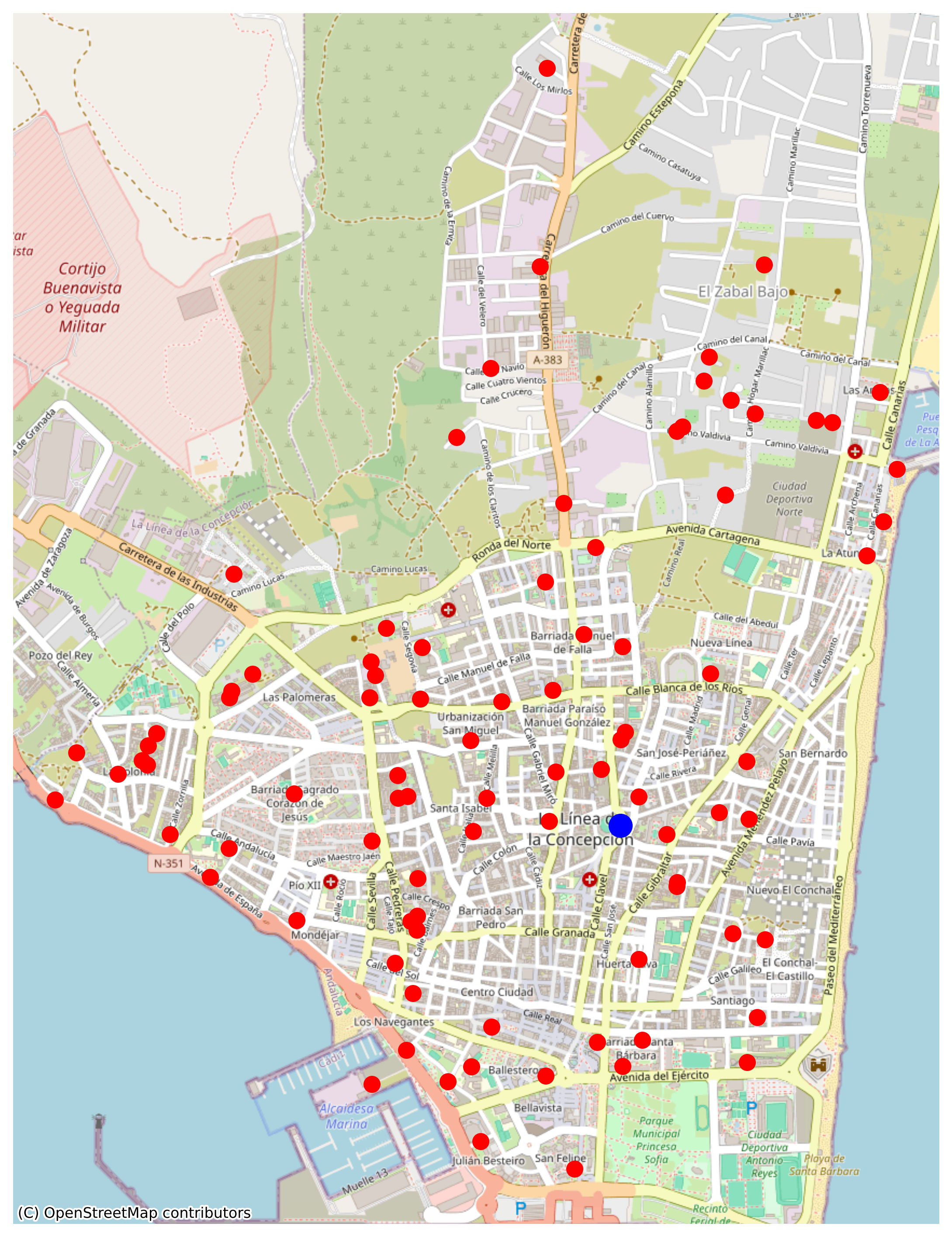}
    \caption{The map of La L\'inea de la Concepci\'on and nodes with highest frequency of appearance. Map data \copyright~OpenStreetMap contributors.}
    \label{fig: lalinea final map}
\end{figure}

\end{document}